\documentclass[11pt]{article}
\usepackage{amsmath,amssymb,amsthm, latexsym,color,epsfig,enumerate,a4, graphicx}

\theoremstyle{plain}
\newtheorem{theorem}{Theorem}[section]
\newtheorem{lemma}[theorem]{Lemma}

\date{}
\title{\vspace{-0.7cm}Strong games played on random graphs}
\author{
Asaf Ferber\thanks{Department of Mathematics, Yale University and
Department of Mathematics, MIT. Emails: {\tt asaf.ferber@yale.edu,
ferbera@mit.edu}.} \and Pascal Pfister \thanks{Institute of
Theoretical Computer Science, ETH Z\"urich, 8092 Z\"urich,
Switzerland. Email: {\tt ppfister@student.ethz.ch}. } }

\begin{document}
\maketitle

\begin{abstract}
In a strong game played on the edge set of a graph $G$ there are two
players, Red and Blue, alternating turns in claiming previously
unclaimed edges of $G$ (with Red playing first). The winner is the
first one to claim all the edges of some target structure (such as a
clique $K_k$, a perfect matching, a Hamilton cycle, etc.). It is
well known that Red can always ensure at least a draw in any strong
game, but finding explicit winning strategies is a difficult and a
quite rare task.

We consider strong games played on the edge set of a random graph $G
\sim G(n,p)$ on $n$ vertices. We prove, for sufficiently large $n$
and a fixed constant $0 < p < 1$, that Red can w.h.p win the perfect
matching game on a random graph $G  \sim G(n,p)$.
\end{abstract}

\section{Introduction}
\emph{Strong games}, as a specific type of \emph{Positional games},
involve two players alternately claiming unoccupied elements of a
set $X$, which is referred to as the \emph{board} of the game. The
two players are called \emph{Red} (the first player) and \emph
{Blue} (the second player). The focus of Red's and Blue's attention
is a given family $\mathcal{H} \subseteq 2^X$ of subsets of $X$,
called the \emph{hypergraph of the game}, or sometimes referred to
as the \emph{winning sets of the game}. The course of the game is
that Red and Blue take turns in claiming previously unclaimed
elements of $X$, \emph{exactly one} element each time, with Red
starting the game. The winner of such a strong game $(X,
\mathcal{H})$ is the first player to claim all elements of some
winning set $F \in \mathcal{H}$. If this has not happened until the
end of the game, i.e. until all elements of $X$ have been claimed by
either Red or Blue, the game is declared as a \emph{draw}.

One classical example of a strong game is the child game Tic-Tac-Toe and its close relative $n$-in-a-row, where the target sets are horizontal, vertical and diagonal lines of a square grid. The player completing a whole line first wins. If at the end of the game no line has been completely claimed by either of the players, then the game is declared as a draw.

Another interesting example is the following generalization of Tic-Tac-Toe -- the $[n]^d$ game. Here, the board is the $d$-dimensional discrete cube $X=[n]^d$, and the winning sets are all the \emph{combinatorial lines} in $X$. Note that, in this notation, the $[3]^2$ game is the familiar Tic-Tac-Toe.

It is also natural to play Positional games on the edge set of a graph $G=(V,E)$. In this case, $X=E$ and the target sets are all the edge sets of subgraphs of $G$ which possess some given graph property $\mathcal P$, such as ``being connected", ``containing a perfect matching", ``admitting a Hamilton cycle", ``being not $k$-colorable", ``containing an isomorphic copy of given graph $H$" etc.

Since a strong game is a finite, perfect information game (as all of the Positional games), a well known fact from Game Theory asserts that, assuming the two players play according to their optimal strategies, the game outcome is determined and it can be in principle: win of Red, win of Blue, or a draw.

In reality however, there are only two possible outcomes for this kind of games (assuming optimal strategies). Applying the so-called \emph{strategy stealing principle}, which was observed by John Nash in 1949, it follows that the first player (Red) cannot lose the game, if she plays according to her optimal
strategy. Hence any strong game, if Red and Blue play according to their optimal strategies, is either \emph{Red's win} or ends in a draw. On one hand, this argument sounds (and indeed is) very general and powerful, but on the other hand, the strategy stealing argument is very inexplicit and gives no clue for how such an optimal strategy for Red looks like.

Another general tool in the theory of strong games are Ramsey-type arguments. They assert that if a hypergraph $\mathcal{H} \subseteq 2^X$ is non-2-colorable (that is, in every coloring of the elements of the board $X$ with two colors, there must exists a monochromatic $F\in \mathcal{H}$), then Red has a winning strategy in the strong
game $(X, \mathcal{H})$. The most striking example of an application of
this method is probably for the above mentioned $[n]^d$ game. Hales and Jewett, in one
of the cornerstone papers of modern Ramsey theory \cite{HJ}, proved that for a given $n$ and a large enough $d\ge d_0(n)$, every 2-coloring of $[n]^d$ contains a monochromatic combinatorial line. Thus, the strong game played on such a board cannot end in a draw and is hence Red's win (but again, no clue how a winning strategy looks
like!).

Regretfully, the above two main tools (strategy stealing, Ramsey-type arguments) exhaust our set of general tools available to handle strong games. In addition, both tools are inexplicit, and Ramsey-type statements frequently provide astronomic bounds. The inherent difficulty in analysing strong games
can be explained partially by the fact that they are not hypergraph monotone. By
this we mean the existence of examples, e.g. provided by J\'ozsef Beck, (Ch. 9.4
of \cite{jozsef2009inevitable}), of game hypergraphs $\mathcal{H}$ which are Red's win, yet one can add an extra set $A$ to $\mathcal{H}$ to obtain a new hypergraph $\mathcal{H}'$ which is a draw. This is what Beck calls the \emph{extra set paradox}, and it is indeed quite disturbing.

Partly due to the great difficulty of studying strong games, \emph{weak games}, also known as \emph{Maker-Breaker games,} were introduced. In the \emph{Maker-Breaker game} $(X,\mathcal{H})$, two players, called Maker and Breaker, take turns in claiming previously unclaimed elements of $X$, with Breaker going first. Each player claims exactly one element of $X$ per turn. Again, the set $X$ is called the \emph{board} of the game and the members of $\mathcal{H}$ are referred to as the \emph{winning sets}. Maker wins the game as soon as she occupies all elements of some winning set $F\in \mathcal{H}$. If Maker does not fully occupy any winning set by the time every board element is claimed by some player, then Breaker wins the game. Note that being the first player is never a disadvantage in a Maker-Breaker game (see e.g. \cite{beck2008combinatorial}). Hence, in order to prove that Maker can win some Maker-Breaker game as the first or second player, it suffices to prove that she can win this game as the second player.

Using fast strategies for Maker-Breaker games (see
\cite{hefetz2009fast}), some very nice and surprising results about
particular strong games played on the edge set of a complete graph
$K_n$ have been obtained recently. The few examples of such strong
games, for which an explicit winning strategy, based on a fast
Maker-Breaker strategy, is known, include the \emph{perfect
matching} game, the \emph{Hamilton cycle} game and the
\emph{k-vertex-connectivity} game (see \cite{ferber2011winning},
\cite{ferber2014weak}), where Red's aim is to build a perfect
matching, a Hamilton cycle, and a $k$-vertex-connected spanning
subgraph of a complete graph $K_n$, respectively.

Since the problem of finding explicit winning strategies for Red is
quite hard, and since there are no general tools for it, it is just
natural to continue exploring such strategies on different type of
boards. Hopefully, at some point a general tool will appear. A very
natural candidate board for playing on is the well known binomial
random graph $G\sim G(n,p)$, where each edge of the complete graph
$K_n$ is being kept with probability $p$, independently at random
(for a very good survey on random graphs the reader is referred to
the excellent book \cite{bollobas2001random}).

In this paper we initiate the study of strong games played on the
edge set of a typical $G\sim G(n,p)$.  In particular, we analyze the
\emph{perfect matching} game played on $G$ and provide Red with a
winning strategy. Here is our main result:

\begin{theorem} \label{thm:perfect_matching}
Let $0<p\leq 1$ be a fixed constant. Then, a graph $G \sim G(n,p)$
is w.h.p such that Red has a winning strategy for the perfect
matching game played on $E(G)$.
\end{theorem}

\subsection{Notation and terminology}

Our graph-theoretic notation is standard and follows that of \cite{west2001introduction}. In particular, we use the following. For a graph $G$, let $V(G)$ and $E(G)$ denote its sets of vertices and edges respectively. Moreover, let  $e(G) := |E|$ be the number of edges of $G$ and, for any two disjonit subset $S,T \subset V(G)$ let $e(S,T)$ be the number of edges with one endpoint in $S$ and the other in $T$. For a set $S\subseteq V(G)$, let $G[S]$ denote the subgraph of $G$, induced on the vertices of S.

Assume that some strong game, played on the edge set of some graph
$G$, is in progress. At any given moment during this game, we denote
the graph spanned by Red's edges by $R$, and the graph spanned by
Blue's edges by $B$. For a set $S\subseteq V(G)$, let $B[S]$,
respectively $R[S]$, denote the subgraph of $B$, respectively of
$R$, induced by the vertices of $S$. At any point during the game,
the vertices of $G\setminus(R\cup B)$ are called \emph{free
vertices} and any edge not yet claimed is called \emph{free edge}.
We also denote by $d_{R}(v)$ and $d_{B}(v)$ the degree of a given
vertex $v\in V(G)$ in $R$ and in $B$ respectively. Moreover, any
vertex $v \in V$ with $d_R(v)=0$ and $d_B(v) >0$ is called
\emph{distinct}.

\section{Preliminaries and tools}

In this section we introduce some tools used in the proof of Theorem \ref{thm:perfect_matching}.

\subsection{Partitioning of $G \sim G(n,p)$}\label{sec:partition}

We first show the following auxiliary theorem and a partitioning lemma for a graph $G \sim G(n,p)$ which will allow Red to partition the board $E(G)$ into suitable subboards.

\begin{theorem} \label{thm:auxiliary}
Let $n$ be a sufficiently large integer and let $0<p\leq 1$ be a fixed constant. Then, w.h.p, a graph $G \sim G(n,p)$ is such that the following holds:\\
There exists a partition $V(G)=V_1 \cup ... \cup V_t$ of $G$ into disjoint subsets such that for all $1\leq i \leq t$ we have:
\begin{description}
\item[$i)$] $G[V_i]$ is a clique
\item[$ii)$] $|V_i| = \Theta(\ln^{\frac{1}{3}}n)$
\item[$iii)$] $|V_i|$ is even for all $1 \leq i \leq t-1$
\end{description}
\end{theorem}

The proof of this theorem closely follows a nice argument of Krivelevich and Patk\'os from the proof of Theorem 1.2 in \cite{krivelevich2009equitable}.

\begin{proof}
We will use the following greedy algorithm: Let $k = \lceil n / \ln^{\frac{1}{3}}n \rceil$ and partition the vertex set $V(G)=U_1 \cup ... \cup U_{\lfloor \frac{n}{k} \rfloor} \cup W$ such that $|U_1|=|U_2|=...= |U_{\lfloor \frac{n}{k} \rfloor}| = k$ and $W = V \setminus \bigcup_{j=1}^{\lfloor \frac{n}{k} \rfloor} U_j$. Note that $|W| \leq k$.

Let $r = \lceil \frac{n}{k} \rceil \leq \ln^{\frac{1}{3}}n$. We
build the $k$ cliques in $r$ rounds by starting with cliques of size
1 and, for all $2 \leq i \leq r-2$,  adding in the  $i^{\emph{th}}$
round one vertex of $U_i$ to each clique. In the last two rounds, we
add the last few vertices ``smartly" to ensure that all but one of
the cliques are of even size. We denote the $k$ cliques obtained
after the $i^{\emph{th}}$ round of the algorithm by
$C^1_i,...,C^k_i$. In the first $(r-2)$ rounds the algorithm works
as follows: In the first round we simply define $ \lbrace
C^1_1,...,C^k_1 \rbrace := U_1$, and hence $C^1_1,...,C^k_1$ is a
collection of $k$ cliques, each of which of size 1. For $2 \leq i
\leq r-2$, in the $i^{\emph{th}}$ round we expose all edges between
$U_i$ and $\bigcup_{j=1}^{i-1}U_j$. To find the extension of the
cliques, we define an auxiliary random bipartite graph $B_i = W_i
\cup U_i$ on $2k$ vertices, where $W_i = \lbrace
C^1_{i-1},...,C^k_{i-1} \rbrace$. That is, $W_i$ represents the
already formed cliques of size $(i-1)$ and the other part stands for
the new vertices we want to add to those cliques. For $C\in W_i$ and
$x\in U_i$, we add the edge $Cx$ to $E(B_i)$ if and only if $x$ is
connected (in $G$) to all the vertices of $C$. Hence, any perfect
matching of $B_i$ corresponds to an extension of the cliques
$C^1_{i-1},...,C^k_{i-1}$ by one vertex each. Note that the
auxiliary graph $B_i$ has edge probability $p_i= p^{i-1}$. It is
shown below that w.h.p there exists a perfect matching in $B_i$ for
all $1 \leq i \leq r-2$. Before that, let us describe the last two
rounds of the algorithm.

Assume that after $(r-2)$ rounds we have $k$ cliques $C^1_{r-2},...,C^k_{r-2}$ of size $(r-2)$. In the last two rounds, the algorithm extends these cliques with the vertices of $U_{\lfloor \frac{n}{k} \rfloor} \cup W$ such that all but at most one of the cliques $C^1_r,...,C^k_r$ are of even size. This is done as follows:\\
If $(r-2)$ is odd, we continue for one more round as described above. Hence we obtain $k$ cliques $C^1_{r-1},...,C^k_{r-1}$ of even size, and we define $L:=W$.\\
Else, if $(r-2)$ is even, we define $L:=U_{\lfloor \frac{n}{k} \rfloor} \cup W$ and update $\lbrace C^1_{r-1},...,C^k_{r-1} \rbrace:= \lbrace C^1_{r-2},...,C^k_{r-2} \rbrace$. Note that $|L| \leq 2k$.

In the last round, the algorithm first partitions $L = X \cup Y$ into two equitable halves (i.e. $||X| -|Y|| \leq 1$) and exposes all the edges in $L$. W.l.o.g. let $|X|\leq |Y|=k' \leq k$ and define the auxiliary random bipartite graph $B_L = X \cup Y$ with edge probability $p_L=p$ (note that we forget about all the edges exposed inside $X$ and $Y$). Let $(x_1,y_1),...,(x_{k'-1},y_{k'-1})$ be the vertices of $X$ and $Y$ which are paired up by a perfect matching in $B_L$. We define the auxiliary set $Z=\lbrace z_1,...,z_{k'} \rbrace$ by $z_i:= (x_i,y_i)$ for all $ 1 \leq i \leq k'-1$ and $z_{k'}=(x_{k'},y_{k'})$ if $|X|=|Y|$, respectively $z_{k'}=y_{k'}$ if $|X|=|Y|-1$.\\
Then, the algorithm exposes all edges between $L$ and $V \setminus L$. We define the auxiliary random bipartite graph $B_{r} = W_{r} \cup Z$ on $2k'$ vertices, with $W_{r} = \lbrace C^1_{r-1},...,C^{k'}_{r-1} \rbrace$. An edge between $z_i$ and $C^j_{r-1}$ is present in $E(B_{r})$ if and only if all edges between $x_i$ and $C^j_{r-1}$ as well as all edges between $y_i$ and $C^j_{r-1}$ are present. Hence we obtain an edge probability $p_{r} = p^{2(r-2)}$ (if $z_{k'} = y_{k'}$, we flip an additional coin with success-probability $q=p^{r-2}$ for the edges connecting $z_{k'}$ with $C^j_{r-1}$ to obtain $p_{r} = p^{2(r-2)}$ also for edges touching $z_{k'}$).\\
If there exists a perfect matching in $B_{r}$, the algorithm extends the cliques $C^1_{r-1},...,C^{k'}_{r-1}$ by the corresponding vertex-pair in $Z$ and therewith obtains $k$ cliques $C^1_r,...,C^k_r$ of which at most one, namely $C^{k'}_r$, is of odd size. Thus, after reordering, the algorithm outputs a partition $V(G)=V_1 \cup ... \cup V_k$ into $k$ disjoint subsets of size $\Theta(\ln^{\frac{1}{3}}n)$, such that $G[V_i]$ is a clique for all $i \in[k]$ and $|V_i|$ is even for all $1 \leq i \leq k-1$.
\medskip

It remains to prove that the algorithm succeeds w.h.p, i.e. that the algorithm can find a perfect matching in the auxiliary bipartite graphs $B_L$ and $B_i$ for all $1 \leq i \leq r$. Since $0<p<1$ is a constant, we have, for all $1 \leq i \leq r$, that $p_L \geq p_i \geq p^{2r} \geq p^{2\ln^{\frac{1}{3}}n} = O(n^{{\alpha}-1})$ for some $0< \alpha < 1$. By Remark 4.3 in Chapter 4 in \cite{janson2011random}, we know that the probability that there is no perfect matching in our auxiliary bipartite graphs $B_i$ is $O(ke^{-kp_i})$. Therewith, the probability that the Algorithm fails is upper bounded by $(r+1) O(ke^{-kp^{2r}}) = O(ne^{-n^{\alpha} / \ln^{\frac{1}{3}}n}) = o(1)$. Thus the algorithm succeeds with high probability and constructs $k$ cliques with the desired properties.
\end{proof}

Using the above theorem, we prove the following lemma which ensures
us a partitioning of a random graph $G \sim G(n,p)$ into disjoint
complete subgraphs which can be cyclicly ordered in such a way that
the union of any two consecutive cliques is a clique as well.

\begin{lemma} \label{partitioning_lemma}
Let $0<p\leq 1$ be a fixed constant. Then, w.h.p, a graph $G \sim G(n,p)$ is such that the following holds:\\
There exists a partition $V(G)=V_1 \cup ... \cup V_t$ of $G$ into disjoint subsets such that for all $1\leq i \leq t$ we have the following:
\begin{description}
\item[$i)$] $G[V_{i}\cup V_{i+1}]$ is a clique (we consider $t+1$ to be 1)
\item[$ii)$] $|V_i| = \Theta(\log^{\frac{1}{3}}n)$
\item[$iii)$] $|V_i|$ is even for all $1 \leq i \leq t-1$.
\end{description}
\end{lemma}

\begin{proof}
Let $q$ be a constant such that $1-p=(1-q)^2$. Present $G=G_1 \cup G_2$, where $G_1,G_2 \sim G(n,q)$ (for details see \cite{bollobas2001random}).\\
Let $V(G_1)=U_1 \cup ... \cup U_l$ be the partition of $G_1$ into disjoint subsets obtained by applying Theorem \ref{thm:auxiliary}. Hence we have that $|U_i|$ is even for all $1 \leq i \leq l-1$.\\
Further, partition each subset $U_i=L_i \cup R_i$ into two halves such that
\begin{description}
\item[$i)$] $||L_i|-|R_i|| \leq 2$ for all $1 \leq i \leq l$
\item[$ii)$] $|L_i|$ and $|R_i|$ are even for all $1 \leq i \leq l-1$
\item[$iii)$]  $|L_l|$ is even.
\end{description}
Hence, all subsets but $R_l$ are of even size and, for all $1 \leq i \leq l$, we have that $L_i$ and $R_i$  are of size  $\Theta(\log^{\frac{1}{3}}n)$.\\
Before exposing $G_2$, define an auxiliary digraph $D=(V,E)$ such that the set of vertices is defined by $V(D) := \lbrace (L_i,R_i) | 1 \leq i \leq l \rbrace$. Furthermore, let $R_iL_j$ be the directed edge from $(L_i,R_i)$ to $(L_j,R_j)$, which is present if and only if all edges between $R_i$ and $L_j$ appear in $G_2$.\\
Note that $\Pr[R_iL_j \in E(D)]=q^{|R_i||L_j|} = q^{\Theta(\log^{\frac{2}{3}}n)} = \omega(\frac{\ln^2 n}{n}) = \omega(\frac{\ln(|V(D)|)}{|V(D)|})$, since $|V(D)|\approx n / \ln^{\frac{1}{3}}n$. Using the main result of \cite{frieze1988algorithm}, we know that the digraph $D$ contains a directed Hamilton cycle. W.l.o.g. let $(L_1,R_1)(L_2,R_2)...(L_l,R_l)$ be this directed Hamilton cycle. By defining $V_1, V_2,...,V_t := U_{1,1},U_{1,2},U_{2,1},U_{2,2},U_{3,1},...,U_{l,2}$ we hence obtain our partition $V(G) = V_1 \cup ... \cup V_t$ with the desired properties.
\end{proof}

\subsection{The perfect matching game on the complete graph $K_n$}

The main tools used in the proof of Theorem
\ref{thm:perfect_matching} are the following two strategies
concerning the perfect matching game on the complete graph $K_n$.
One is the strategy described in the proof of Theroem 1.2 in
\cite{hefetz2009fast} which ensures that Maker can win the weak
perfect matching game on $K_n$ in at most $n/2+1$. We will
henceforth denote this strategy by $\mathcal{S}^{weak}_n$. The
second strategy is a slight alteration of the strategy which ensures
that Red can win the strong perfect matching game on $K_n$, as
described in the proof of Theorem 1.3 in \cite{ferber2011winning}.
Before describing this strategy we need the following definitions.
For any matching $M \in G$, let $e(M)$ be the number of edges in this matching and let $M_G := \max \lbrace e(M) : M \subset G \emph{ is a matching in G} \rbrace$ denote the size of a maximum matching in G. When a strong perfect matching game is in progress, we say that Blue (respectively Red) \emph{wastes} a move, if she claims an edge which does not increase $M_B$ (respectively $M_R$). Note that the game we propose below can be thought of as an ``almost strong" perfect matching game, because Red's strategy gives her a perfect matching fast (in at most $n/2 +2$ moves), without wasting more moves than Blue. But since Blue may have already claimed an edge on the board before Red starts to play, the strategy cannot assure that Red builds a perfect matching before Blue does (hence the ``almost strong"). Now we are ready to state and prove the following:
\begin{lemma} \label{lemma:strategy_strong}
Let $H = K_n$ and let $G \supseteq H$ be a graph on $n' \geq n$ vertices. Assume that, when Red starts claiming edges, there exists a vertex $v \in V(H)$ with $d_B(v) \geq 1$, but Blue claimed at most one edge $xy$ in $E(H)$. Then Red can build a perfect matching on $H$ in at most $n/2+2$ moves. Moreover, Red will not waste more moves than Blue.
\end{lemma}

The proof of Lemma \ref{lemma:strategy_strong} more or less follows
the lines of Theorem 1.3 in \cite{ferber2011winning}.

\begin{proof} Red's goal in this ``almost strong" perfect matching game is to build a perfect matching on a complete subgraph $H \subseteq G$ in at most $n/2 +2$ moves while not wasting more moves than Blue does.
In what follows, we present a strategy for Red and then prove that,
by following it, Red can build a perfect matching on $H$ in at most
$n/2+2$ moves while not wasting more moves than Blue.
\medskip

Assume first that $n$ is odd. Following Maker’s strategy $\mathcal{S}^{weak}_n$ on $E(H)$, Red can build an almost perfect matching on $H$ in $\lfloor n/2 \rfloor$ moves. Hence, if $n$ is odd, then Red plays according to $\mathcal{S}^{weak}_n$ on $E(H)$.
\medskip

Else, Red's strategy is divided into the following three stages:
\medskip

\textbf{Stage I:} In her first move, depending on whether Blue
claimed an edge in $E(H)$ or not, Red distinguishes between the
following two cases:
\begin{description}
\item[Case 1:] Blue claimed an edge $xy \in E(H)$.\\
Red then claims a free edge $xz$ for some arbitrary $z \neq y \in
V(H)$, defines the set  $U:=V(H)$ and skips to Stage II.
\item[Case 2:] $e(B[H]) = 0$, but there exists a vertex $u \in H$ with $d_B(u) \geq 1$.\\
Red defines the set $U:=V(H) \setminus \lbrace u \rbrace$, claims an arbitrary free edge in $E(G[U])$ and skips to Stage II.
\end{description}

To describe Stage II, we define for all vertices $v \in V(H)$ the \emph{H-degree} $d_B^H (v)$ of Blue, respectively $d_R^{H}(w)$ of Red, as the number of edges connecting $v$ to other vertices $w \in V(H)$ in Blue's graph, respectively in Red's graph. Moreover, a vertex $v \in V(H)$ is called \emph{H-distinct}, if $d_B^{H}(v) \geq 1$ and $d_R^{H}(v)=0$. Hence let $D_j$ be the number of H-distinct vertices immediately \emph{after} Red's $j^{\emph{th}}$ move. Additionally, let ${D'}_j$ be the number of H-distinct vertices immediately \emph{before} Red's $j^{\emph{th}}$ move.
\medskip

\textbf{Stage II:} For every $2 \leq j \leq n/4+2$, in her $j^{\emph{th}}$ move Red claims an edge $e_j \in E(G[U])$ which is independent of her previously claimed edges while making sure that $D_j \leq 1$. Red can even ensure that, if $D_k = 1$ for some $1 \leq k \leq n/2-1$, then $D_j = 1$ for all $k \leq j \leq n/2-1$ (we will prove later that this is indeed possible). Hence, let $2 \leq k \leq n/2-1$ be the smallest integer such that $D_k = 1$. Then Red updates $U := H$ in her $k^{\emph{th}}$ move, since she does no longer need the ``trap vertex" $u \in V(H)$.\\
If $\Delta(B[H]) > 1$ holds immediately after Blue’s $(n/4+2)^{\emph{nd}}$ move, then Red skips to Stage M. Otherwise, for every $n/4 + 3 \leq j \leq n/2-1$, in her $j^{\emph{th}}$ move Red claims an edge $e_j \in E(G[U])$ which is independent of her previously claimed edges while making sure that $D_j \leq 1$. Red then proceeds to Stage III.
\medskip

\textbf{Stage III:} Red completes her perfect matching in $E(H)$ by claiming at most 3 additional edges as follows:\\
Let $x,y \in V(H)$ be the two last vertices Red needs to connect to
build a perfect matching on $E(H)$. In her $(n/2)^{\emph{nd}}$ move,
Red claims $xy$ and finishes her perfect matching in $E(H)$. If this
is not possible, let $uv$ and $wz$ be two edges in $E(H)$ such that
$B[\lbrace u, v, w, z, x, y \rbrace]$ consists solely of the edge
$xy$. In her $(n/2)^{\emph{nd}}$ move Red then claims the edge $yu$.
In her $(n/2+1)^{\emph{st}}$ move, Red then claims the edge $xv$ and
thus finishes her perfect matching in $E(H)$. If this is not
possible, Red claims the edge $xz$. Since Blue cannot claim both
$wy$ and $wv$ in her next move, Red claims one of them in her
$(n/2+2)^{\emph{nd}}$ move and thus finishes her perfect matching in
$E(H)$ wasting at most two moves.
\medskip

\textbf{Stage M:} Let  $I_H := \lbrace v \in V(H) \ | d_R(V)=0 \rbrace$ be the set of isolated vertices of Red in $V(H)$. Note that $|I_H|=n/2-4$ is even. Playing on $E(G[I_H])$, Red follows the strategy $\mathcal{S}^{weak}_{n/2-4}$.
\bigskip

It remains to prove that Red can indeed follow all parts of the strategy.
\medskip

For Case 1 of Stage I note that Red uses the vertex $u$ as ``trap vertex", since Blue wastes a move by touching it again (because only one of the two (or more) edges incident to $u$ in Blue's graph are in the same maximum matching). Hence Red needs to ensure, as long as $D_j =0$, that the last edge she needs for her perfect matching is incident to $u$.
Furthermore, after Stage I, we have that $D_1 \leq 1$.
\medskip

The following lemma asserts that Red can follow Stage II of her strategy (either
for $n/4 +2$ or $n/2-1$ moves).

\begin{lemma} \label{lemma:distinct}
Let $H = K_n$ and let $G \supseteq H$ be a graph on $n' \geq n$ vertices. Assume that, when Red starts claiming edges, there exists a vertex $v \in V(H)$ with $d_B(v) \geq 1$, but Blue claimed at most one edge $xy$ in $E(H)$. Then Red can ensure that, for all $1 \leq j \leq n/2-1$, immediately after her $j^{\emph{th}}$ move, her graph is a matching consisting of $j$ edges and $D_j \leq 1$. Moreover, if $D_k = 1$ for any $1 \leq k \leq n/2-1$, then $D_j = 1$ for all $k \leq j \leq n/2-1$.
\end{lemma}

\begin{proof}
We prove the lemma by induction on $j$. Stage I ensures that $D_1 \leq 1$. Note that, since Blue can create at most two H-distinct vertices in one round, we have that $D'_{j+1} - D_j \leq 2$. We distinguish now two cases:

\begin{description}
\item[Case 1:] $D_j = 1$.\\
If $D'_{j+1} = 1$, let $u \in V(H)$ be the H-distinct vertex with $d^{H}_B(u) \geq 1$ and $d^{H}_R(u)=0$. Then Red claims any free edge $xy \in E(G[V(H) \setminus \lbrace u \rbrace])$ which is independent of all her previously claimed edges in $E(H)$ and hence $D_j=D_{j+1}=1$.\\
Else, if $D'_{j+1} = 2$, let $u \neq w \in V(H)$ be the two H-distinct vertices. Then Red claims an arbitrary free edge $wx$ with $x \neq u$ in $E(H)$ which is independent of all her previously claimed edges and hence $D_j=D_{j+1}=1$.\\
Else, if $D'_{j+1} = 3$, let $u \neq w \neq z \in V(H)$ be the three H-distinct vertices. Then w.l.o.g. let $wz$ be the edge Blue claimed in her last move. This means that $uw$ and $uz$ are still free, since there was only one H-distinct vertex before Blue's move. Hence Red claims one of them, thus ensuring that $D_j=D_{j+1}=1$.
\item[Case 2:] $D_j = 0$.\\
If $D'_{j+1} = 0$, then Red claims any free edge $uv \in E(G[U])$ (note that this is the only case where the auxiliary vertex-set $U$ is not equal to $V(H)$), which is independent of all her previously claimed edges in $E(H)$ and hence $D_j=D_{j+1}=0$.\\
Else, if $D'_{j+1} = 1$, let $u \in V(H)$ be the H-distinct vertex. Then Red claims an arbitrary free edge $xy \in E(G[V(H) \setminus \lbrace u \rbrace])$ which is independent of all her previously claimed edges in $E(H)$ and hence $D_{j+1}=1$.\\
Else, if $D'_{j+1} = 2$, let $u \neq w \in V(H)$ be the two H-distinct vertices. Then Red claims an arbitrary free edge $wx$ with $x \neq u$ in $E(H)$ which is independent of all her previously claimed edges and hence $D_{j+1}=1$.
\end{description}

Note that with keeping $D_j \leq 1$, Red ensures that there are always enough vertices $v \in V(H)$ with $d_B^H(v)=d_R^H(v)=0$. Therefore those independent edges which Red claims always exist.
\end{proof}

Lemma \ref{lemma:distinct} assures that Red can follow Stage II of the strategy.
\medskip

When Red reaches Stage III, Lemma \ref{lemma:distinct} ensures that Red's graph consists of a matching with $n/2 - 1$ edges. Note that $\Delta(B[H]) \leq n/4-3$, and therefore $d_B^H(x) \leq n/4 -3$. Moreover, remember that $d_B^H(y)=0$. Hence such two edges $uv$ and $wz$ will always exist in $E(H)$, since Blue cannot connect $x$ to more than $n/4-3$ edges of Red's matching.\\
Since $d_B(x) \geq 1$ before Red's $(n/2)^{\emph{th}}$ Move, Blue wastes a move by claiming $xy$ (and $xv$), again because only one of the two (or three) edges incident to $x$ in Blue's graph are in the same maximum matching. Thus, after Stage III, Red built a perfect matching in $E(H)$ in at most $n/2+2$ moves and, additionally, wasted at most as many moves as Blue.
\medskip

When Red reaches Stage M, $\Delta (B[H]) > 1$ and hence Blue wasted at least one move. Hence Red might waste one move too and therewith can play according to Maker's strategy on $E(G[I_H])$. Thus she builds her perfect matching in $E(H)$ in $n/2+1$ moves and, additionally, does not waste more moves than Blue.
\end{proof}

Henceforth, the strategy described in the proof above will be denoted by $\mathcal{S}^{a.strong}_n$.

\section{Proof of Theorem \ref{thm:perfect_matching}}

The main idea of Red's strategy is to build a perfect matching on $G
\sim G(n,p)$ ``quickly" while ensuring that she does not waste more
moves than Blue. To this end, Red partitions $V(G)$ into suitable,
cyclically ordered subgraphs as obtained by applying Lemma
\ref{partitioning_lemma} to $G$. In her strategy, Red then mostly
neglects all the edges in between those ``subboards" and plays on each
subboard seperately, trying to complete a perfect matching on each
board in a cyclically order. The crucial observation here is that
whenever Blue blocks the last edge Red needs for a perfect matching
on a subboard, using the fact that the union of two consecutive
subboards is a clique as well, Red can ``import" two vertices from
the next subboard to circumvent Blue's attack. Red is only being
interrupted in this ``subboard by subboard" approach if Blue tries
to block a vertex or claims too many edges on a subboard which Red
has not been playing on yet. Whenever a certain amount of edges in a
specific subboard or edges incident with the same vertex is reached
in Blue's graph, Red marks the relevant subboard as ``dangerous" and
gives this board a special attention.\\
Another thing Red needs to be carefull about is, that when she reaches a subboard to play on, it might not be `empty", since Blue cuold already have claimed some edges in it. But this is not really a problem, since any distinct vertex gives Red an advantage, because, usually, Blue may not touch distinct vertices again without wasting a move. Thus any disitnct vertex $v$ on an empty suboard can be used as a ``trap vertex" by Red if she ensures that the last edge of her perfect matching on this subboard will be incident to $v$.
\medskip

For the description of Red's strategy, we will use the following notation and definitions:\\
Assume that the graph $G \sim G(n,p)$ is partitioned according to Lemma \ref{partitioning_lemma}. Hence we have a partition $V(G)=V_{1}\cup\ldots\cup V_{t}$ into $t$ disjoint subsets. For all $1 \leq i \leq t$, let $E_i := E(G[V_i])$ be the subboards Red will play on. At any point during the game, let $R_i \subset R$ be the subgraph of Red's graph induced on $V_i$ and let $B_i \subset B$ be the subgraph of Blue's graph induced on $V_i$. Moreover, a subboard $E_i$ is called \emph{inactive} if $e(R_i) = 0$, it is called \emph{active} if $e(R_i) \geq 1$, but $R_i$ does not contain a perfect matching, and it is called \emph{safe} if $R_i$ contains a perfect matching.\\
Assume that a strong perfect matching game is in progress.  As defined before, a wasted move of Blue, respectively Red, is the claiming of an edge which does not increase the size of a maximum matching in $B$, respectively in $R$. However, since Red will mostly claim edges inside a
subboard $E_i$ during the game, and the tracking of wasted edges which lie between subboards, not all (wasted) moves of Blue do
concern Red. Thus, a \emph{wasted move of Blue on the
subboard $E_i$} is the act of claiming of an edge in $E_i$ which does not increase the size $M_{B_i}$ of a maximum matching in $B_i$. When we use the term \emph{wasted move} afterwards, we reffer to wasted moves in a specific subboard, unkess noted otherwise. Furthermore, we define the function
$w:\lbrace 1,...,t \rbrace \rightarrow  \lbrace 0, 1 \rbrace$ by $w(i) =0$ if $e(B_i)/2 - M_{B_i} = 0$ and $w(i)=1$ otherwise.\\
During the game, to prevent Blue from blocking a vertex or claiming
too many edges on an inactive subboard $E_i$, Red keeps track of $w(i)$. If a subboard $E_i$ is inactive and $w(i)$ turns 1, meaning that Blue wasted her first move on the subboard $E_i$, this subboard becomes \emph{dangerous}. From the moment on that a
subboard becomes dangerous, whenever Blue claims an edge in such a
board, Red skips her ``subboard by subboard" approach and answers on
the same board. Note that whenever an inactive subboard $E_i$
becomes dangerous, it stays dangerous until Red has completed a
perfect matching in $R_i$.
\medskip

\begin{proof}[Proof of Theorem \ref{thm:perfect_matching}]

First we propose a strategy for Red and then prove that the proposed
strategy is indeed a winning strategy for the perfect matching game
played on a typical $G\sim G(n,p)$. From now on, we condition on $G$ satisfying all the properties mentioned in the statements at Section \ref{sec:partition}. Hence let  $V(G)=V_{1}\cup\ldots\cup V_{t}$ be the partition of $G\sim G(n,p)$ into $t$ disjoint subsets as described in Lemma \ref{partitioning_lemma}. For all $1 \leq i \leq t$, let $E_i := E(G[V_i])$ and define $m_i := |V_i|$.
\medskip

The description of Red's strategy consists of two parts. In the first part which is called the \emph{overall strategy}, we describe how Red plays based on three ``substrategies" $\mathcal{S}_{dangerous}$, $\mathcal{S}_{trap}$ and $\mathcal{S}_{empty}$, which are given there as black boxes.
In the second part we describe each of these ``substrategies" formally.
Roughly speaking, all these ``substrategies" rely on $\mathcal{S}^{weak}_n$-- which ensures that Maker can win the weak perfect matching game on $K_n$ in at
most $n/2+1$ moves, and the strategy $\mathcal{S}^{a.strong}_n$ -- as described in the proof of Lemma \ref{lemma:strategy_strong}.
\medskip

In her Strategy, except of the first move, Red always reacts to Blue's moves.
Hence, we consider one \emph{round} as one move of Blue and a countermove of Red (except of round 0, which only consists of the first move of Red).
At any point during the game, if Red is unable to follow the proposed strategy or if Red claimed more than $n/2 + 4t$ edges (where $t$ is the number of subboards obtained by partitioning $V(G)$), then she forfeits the game. Red's strategy is divided into the following parts:
\bigskip

\noindent \textbf{Overall Strategy:}
As long as Red's graph does not contain a perfect matching, Red does the following:\\
In her first move of the game Red claims an edge $ab \in E_1$ obtained by following $\mathcal{S}_{empty}$ on $E_1$. In any other move,
let $e_j$ be the edge which has just been claimed by Blue in her $j^ {\emph{th}}$ move.
First, Red checks whether for some $i\leq t$ we have $e_j \in E_i$ and $E_i$ is dangerous. If this is the case, Red answers according to the strategy $\mathcal{S}_{dangerous}$ on $E_i$.
Otherwise, if there exists some active subboard, then Red chooses the smallest integer $i \leq t$ for which $E_i$ is active and plays on $E_i$ according to her chosen strategy $\mathcal{S}_i$ (it is described below how Red chooses the strategy $\mathcal{S}_i$ on each suboard $E_i$).
Otherwise, there are no active subboards and Red's graph contains a perfect matching on the subboards $E_1,...,E_{k-1}$ for some $k \in [t]$.
In this case Red wants to play on $E_k$ according to the strategy $\mathcal{S}_k$, which she chooses in the following way:
\begin{description}
\item[Case 1:] $2 \leq k \leq t-1$.\\
If $e(B_i) \geq 1$, then Red defines $\mathcal{S}_k = \mathcal{S}_{trap}$.\\
Else, $e(B_i) = 0$ and then Red defines $\mathcal{S}_k = \mathcal{S}_{empty}$.
\item[Case 2:] $k=t$.\\
If there exists a vertex $u \in V_k$ with $d_B(u) \geq 1$, then Red defines $\mathcal{S}_k = \mathcal{S}_{trap}$.\\
However, if $d_B(v) =0$ for all $v \in V_k$, we need to distinguish between the following two subcases:
\begin{description}
\item[Case 2.1:] Blue's graph does not contain a perfect matching of $G[V \setminus V_t]$. In this subcase
Red plays according to the Maker-Breaker strategy $\mathcal{S}^{weak}_{m_t}$ on $E_t$.
\item[Case 2.2:] Blue's graph contains a perfect matching of $G[V \setminus V_t]$. In this subcase
Red claims an arbitrary edge $ab \in E_t$ in her first move on $E_t$. If Blue claims an edge $xy \in E_t$ in her first move on $E_t$, Red skips to Stage II of $\mathcal{S}^{a.strong}_{m_t}$ and finishes her perfect matching accordingly.\\
Else, Red plays according to $\mathcal{S}^{weak}_{m_t-2}$ on $E(G[V_t \setminus\lbrace a,b \rbrace])$.
\end{description}
\end{description}
\bigskip

We now give a formal and detailed description of the ``substrategies" to be used by Red.\\
\noindent \textbf{\boldmath $\mathcal{S}_{dangerous}$:}
The strategy $\mathcal{S}_{dangerous}$ is used to play on subboards $E_i$ with $e(B_i) > 1$ and goes as follows:\\
If $B_i$ consists only of two incident edges, say $xy$ and $yz$, then Red claims $xz$, skips to Stage II of the strategy $\mathcal{S}^{a.strong}_{m_i}$ and finishes her perfect matching on $E_i$ according to $\mathcal{S}^{a.strong}_{m_i}$.\\
Else, Red partitions the vertex set $V_i=U_i \cup W_i$ into two subsets, such that
\begin{description}
\item[$(i)$] $|U_i|$ and $|W_i|$ are even if $|V_i|$ is even, and
\item[$(ii)$] $||U_i| - |W_i|| \leq 2$, and
\item[$(iii)$] Blue's graph contains at most one edge in $E(G[U_i]) \cup E(G[W_i])$
\end{description}
(We will prove bellow that this is indeed possible).\\
If there exists an edge in $E(G[U_i]) \cup E(G[W_i])$, we may assume
without loss of generality that it belongs to $E(G[U_i])$. In her first move on $E_i$, Red plays according to $\mathcal{S}^{a.strong}_{|U_i|}$ on $E(G[U_i])$. After this first move, let $u \in U_i$ be either the H-distinct vertex in $U_i$ or the trap vertex Red chose just now. Before her next move, Red then chooses a vertex $w \in W_i$ such that $uw \notin E(B)$ and fixes $w$ as the trap vertex to be chosen when playing according to $\mathcal{S}^{a.strong}_{|W_i|}$ on $E(G[W_i])$.\\
Then, as long as Red's graph does not contain a perfect matching on $E_i$, Red plays as follows:\\
If the last edge Blue claimed was in $E(G[U_i])$, Red plays according to $\mathcal{S}^{a.strong}_{|U_i|}$ on $E(G[U_i])$.\\
Else, if the last edge Blue claimed was in $E(G[W_i])$, Red plays according to $\mathcal{S}^{a.strong}_{|W_i|}$ on $E(G[W_i])$.\\
Else, if Red's graph does not contain a perfect matching on $E(G[U_i])$, Red plays according to $\mathcal{S}^{a.strong}_{|U_i|}$ on $E(G[U_i])$.\\
Else, Red plays according to $\mathcal{S}^{a.strong}_{|W_i|}$ on $E(G[W_i])$.

\bigskip
\noindent \textbf{\boldmath $\mathcal{S}_{trap}$:}
The strategy $\mathcal{S}_{trap}$ is used on subboards $E_i$ where $e(B_i) \geq 1$, respectively on the last subboard $E_t$ if  there exists a ``trap vertex" $u \in V_i$ with $d_B(u) \geq 1$. It can be divided into the following two cases:
\begin{description}
\item[Case 1:] $e(B_i) \leq 1$.\\
Then Red plays according to $\mathcal{S}^{a.strong}_{m_i}$ on $E_i$.
\item[Case 2:] $e(B_i) > 1$.\\
Then Red plays according to $\mathcal{S}_{dangerous}$ on $E_i$.
\end{description}

\bigskip
\noindent \textbf{ \boldmath $\mathcal{S}_{empty}$:}
The strategy $\mathcal{S}_{empty}$ is used on subboards $E_i$ where $e(B_i) = 0$. It consists of the following stages:
\medskip

\textbf{Stage I:} Red claims an arbitrary edge $ab \in E_i$, defines $U := V \setminus \lbrace a,b \rbrace$ and skips to Stage II.
\medskip

\textbf{Stage II:} Red follows Stage II (and Stage M, if needed) of the strategy $\mathcal{S}^{a.strong}_{m_i}$. Then Red skips to Stage III below.
\medskip

\textbf{Stage III:} Red completes her perfect matching on $E_i$ as follows:\\
Let $c,d \in V_i$ be the last two vertices Red needs to connect to finish her perfect matching on $V_i$. Red then claims $cd$. If this is not possible, then we distinguish between two cases:
\begin{description}\item[Case 1:] $e(R_{i+1}) \geq 1$. \\
Red plays as described in Stage III of the strategy $\mathcal{S}^{a.strong}_{m_i}$ and finishes her perfect matching in at most three moves.
\item[Case 2:] $e(R_{i+1}) = 0$.\\
Red chooses two vertices $p,q \in V_{i+1}$ such that
\begin{description}
\item[$(i)$] the edges $cp$ and $dq$ are still free and
\item[$(ii)$] $e(q,V_i) < m_i / 4$.
\end{description}
Then Red claims the edge $cp$. If this is not possible, Red uses Stage III of the strategy $\mathcal{S}^{a.strong}_{m_i}$ to finish her perfect matching on $V_i$ in at most three moves.\\
Before her next move, Red updates $V_i := V_i \cup \lbrace p,q \rbrace$ and $V_{i+1} := V_{i+1} \setminus \lbrace p,q \rbrace$. Then Red claims $dq$. If this is not possible, Red uses Stage III of the strategy $\mathcal{S}^{a.strong}_{m_i}$ to finish her perfect matching on the updated subset $V_i$ in at most three moves.
\end{description}
\bigskip
It remains to prove that Red can indeed follow all parts of the overall strategy as well as the three  ``substrategies",
and that this ensures her win in the strong game of building a perfect matching on $G$.

\bigskip
\textbf{ \boldmath $\mathcal{S}_{dangerous}$:}
Observe that when $B_i$ consists only of a path $xyz$ of length 2, Red can claim $xz$, hence ensuring that there is only 1 distinct vertex in $V_i$. Therefore he can immediately skip to the Strategy $\mathcal{S}^{a.strong}_{m_i}$.\\
For the partition of $V_i$ note that when $w(i)$ turns 1, $B_i$ consists of a path of length 2 or 3 and perhaps some single edges. Therefore Red can ``dissect" all edges of Blue and, by distributing free vertices equally, obtain a partition of $V_i$ into two equitable halves $U_i$ and $W_i$. If $|V_i|$ is not divisible by 4, Red might not be able to dissect all edges in order to keep $|U_i|$ and $|W_i|$ even. But then only one edge of Blue remains in $E(G[U_i]) \cup E(G[W_i])$, and hence the described partitioning of $V_i$ is possible.
Furthermore, this partiton ensures that $E(G[U_i])$ and $E(G[W_i])$ have the right properties to start the strategies $\mathcal{S}^{a.strong}_{|U_i|}$ and $\mathcal{S}^{a.strong}_{|W_i|}$. Hence Red builds a perfect matching on $E(G[U_i])$ at most $|U_i|/2+2$ and on $E(G[W_i])$ in at most $|W_i|/2+2$ moves. Thus, Red finishes her perfect matching on $E_i$ in at most $m_i/2+4$ moves.\\
Moreover, since the two``trap vertices" $u$ and $w$ are not adjacent, any move of Blue blocking the last edge Red needs to finish her perfect macthing on $U_i$ or $W_i$ is indeed a wasted move. Note that, if $u$ and $w$ were adjacent, this would not be true. Because then Blue could block the last edge on both $U_i$ and $W_i$, thereby creating a path of length 3 with the edge $uw$ in the middle. And this could only be one wasted move. Hence Red would waste one move more than Blue by circumventing the two blocked edges using Stage III of the strategy $\mathcal{S}^{a.strong}_{m_i}$. By ensuring that $u$ and $w$ are not adjacent, Red makes sure to not waste more moves than Blue in the building of the perfect matching on $V_i$. Therefore Red still has one spare move left to waste ($w(i) = 1$), which she might use to finish her perfect matching on $E_{i-1}$.\\
Furthermore, apart from the last subboard $E_t$, Red only considered edges in $E_i$ to detect wasted moves of Blue.

\bigskip
\textbf{ \boldmath $\mathcal{S}_{trap}$:}
If Blue claimed one or no edge in $E_i$, Red can immediately start the strategy $\mathcal{S}^{a.strong}_{m_i}$. Otherwise, the partition is needed to ensure that $E(G[U_i])$ and $E(G[W_i])$ have the right properties to start the strategies $\mathcal{S}^{a.strong}_{|U_i|}$ and $\mathcal{S}^{a.strong}_{|W_i|}$. Note that $w(i) = 0$, which means that $B_i$ consists of single edges only, and therefore the board $E_i$ can be partitioned accordingly. Moreover, if Red partitions $V_i$, Blue claimed at least two single edges in $E_i$ and hence Red can choose two trap vertices in $U_i$ and $W_i$ which are not adjacent. Thus, Red finishes her perfect matching on $E_i$ in at most $m_i/2+4$ moves. Moreover, Lemma \ref{lemma:strategy_strong} and $\mathcal{S}_{dangerous}$ ensure that Red does not waste more moves on $E_i$ than Blue.\\
Furthermore, apart from the last subboard $E_t$, Red only considered edges in $E_i$ to detect wasted moves of Blue.

\bigskip
\textbf{ \boldmath $\mathcal{S}_{empty}$:}
First of all, note that Red will not play according to $\mathcal{S}_{empty}$ on the last board $E_t$, hence the trick with importing two vertices from $V_{i+1}$ works if needed.\\
For Stage II note that $D_1= 0$, and hence the basic step of the induction in Lemma \ref{lemma:distinct} is satisfied. Thus Red can skip to Stage II of the strategy $\mathcal{S}^{a.strong}_{m_i}$.\\
For Stage III note that in Case 1, since Red's ''subboard by subboard" approach, $e(R_{i+1}) \geq 1$ means that $E_{i+1}$ is, respectively was, dangerous. Therewith, Blue wastes at least one move more on $E_{i+1}$ than Red.  Hence Red might waste this spare wasted move from the dangerous subboard on $E_i$. She needs it to ``start" with Stage III of $\mathcal{S}^{a.strong}_{m_i}$, since Blue may not waste a move by claiming the last edge $cd$. But with this spare wasted move Red ensures that, overall, she does not waste more moves than Blue.\\
For Case 2 note that, if such two vertices $p$ and $q$ would not exist, $c,d$ or $q$ would have a very high degree in Blue's graph (remember that all edges between $V_i$ and $V_{i+1}$ are present). Hence Red then could finish according to Stage III of $\mathcal{S}^{a.strong}_{m_i}$ and waste one (or two) moves.\\
Moreover, if Blue claims $cd$ and $dq$, she would waste a move, and then Red could finish according to Stage III of $\mathcal{S}^{a.strong}_{m_i}$. Note that here we need $e(q,V_i) < m_i / 4$ to ensure the existence of the set $\lbrace d=x,q=y,u,v,w,z \rbrace$ used in Stage III of $\mathcal{S}^{a.strong}_{m_i}$. Thus Red builds a perfect matching on $E_i$ in at most $m_i/2+2$ moves, and Lemma \ref{lemma:strategy_strong} ensures that Red does not waste more moves on $E_i$ than Blue.\\
Furthermore, Red only considered edges in $E_i$ to detect wasted moves of Blue.

\bigskip
\textbf{Overall Strategy:}
The overall strategy considers all possible cases since either Blue already claimed an edge in $E_i$ or not when Red needs to choose how to play on $E_i$.\\
In general, observe that the three substrategies work in such a way, that, on any subboard, Red wastes not more moves than Blue. Hence, if Red manages to build a perfect matching on $G$, then trivially she does it first and therefore wins the game.\\
Moreover, notice that, on the first $t-1$ subboards, Red never considered edges of Blue between subboards (after updating) to detect wasted moves. Hence on the last subboard, Red can already use $\mathcal{S}^{trap}$ if there exists a vertex $u in V_t$ with $d_B(u) \geq 1$ and use it as a trap vertex. No edge insident to $u$ was considered before to detect wasted moves and therefore claiming an edge incident to $u$ will be a ``new" wasted move of Blue.\\
For Case 2.1 note that Blue needs at least $m_t/2+1$ moves to finish her perfect matching on $G$. Hence, Red can play according to Maker's strategy $\mathcal{S}^{weak}_{m_t}$ on $E_t$, waste one additional move and still finish her perfect matching before Blue does.\\
For Case 2.2 note that if Blue's graph contains a perfect matching of $G[V \setminus V_t]$, then any edge not in $E_t$ which Blue claims can not increase the size of Blue's maximum matching and hence is a wasted move. Then Red can play according to Maker's strategy on $E(G[V_t \setminus \lbrace a,b \rbrace])$ and thus finish her perfect matching on $E_t$ in $m_t/2$ moves (since she already claimed the edge $ab \in E_t$).
Moreover, if Blue claims an edge $xy \in E_t$, then we are in exactly the same scenario as of Stage I of the strategy $\mathcal{S}^{a.strong}_{m_t}$ and hence Red skips
to Stage II of $\mathcal{S}^{a.strong}_{m_t}$ .
\medskip

All in all, since $m_i$ is even for all $1 \leq i \leq t-1$, and since Red can build a perfect matching on all subboards $E_i$ in at most $m_i / 2 + 4$ moves, Red builds a perfect matching on $G \sim G(n,p)$ in at most $n/2 + 4t$ moves. Furthermore, while playing according to the above strategy, Red does not waste more moves than Blue, and thus Red wins the strong perfect matching game on $G \sim G(n,p)$.
\end{proof}

\section{Concluding remarks}

In this paper we considered the perfect matching game played on the
edge set of a typical $G\sim G(n,p)$, for any fixed constant $p$.
Since a perfect matching appears in a typical $G\sim G(n,p)$ when
$p\geq\frac{\ln n+\omega(1)}{n}$ (see e.g.,
\cite{bollobas2001random}), it will be interesting to extend our
result for every $p$ in this regime. Clearly, our proof technique,
which is based on the existence of large cliques, can not work for
small $p$-s.

It might be very interesting to analyze other games as well. For
example, a natural game to analyze is the Hamiltonicity game. Using
similar arguments we indeed managed to provide Red with a winning
strategy for the Hamiltonicity game played on the edge set of a
typical $G\sim G(n,p)$ where $p$ is constant. However, the proof is
quite long and technical so we will only upload it to arXiv as a
draft for the curious reader.

{\bf Acknowledgment:} Most of this work has been taken in ETH, Zurich as part of a master thesis of the second author, while the first author was a postdoc there. We would like to thank ETH, Zurich for providing us with a perfect environment for doing research.
\newpage
\bibliographystyle{plain}
\bibliography{StrongGamesonGnpPaper}

\begin{thebibliography}{10}

\bibitem{beck2008combinatorial}
J{\'o}zsef Beck.
\newblock {\em Combinatorial games: tic-tac-toe theory}.
\newblock Number 114. Cambridge University Press, 2008.

\bibitem{jozsef2009inevitable}
J{\'o}zsef Beck.
\newblock {\em Inevitable randomness in discrete mathematics}, volume~49.
\newblock American Mathematical Soc., 2009.

\bibitem{bollobas2001random}
B{\'e}la Bollob{\'a}s.
\newblock Random graphs. 2001.
\newblock {\em Cambridge Stud. Adv. Math}, 2001.

\bibitem{ferber2011winning}
Asaf Ferber and Dan Hefetz.
\newblock Winning strong games through fast strategies for weak games.
\newblock {\em the electronic journal of combinatorics}, 18(1):P144, 2011.

\bibitem{ferber2014weak}
Asaf Ferber and Dan Hefetz.
\newblock Weak and strong k-connectivity games.
\newblock {\em European Journal of Combinatorics}, 35:169--183, 2014.

\bibitem{frieze1988algorithm}
Alan~M. Frieze.
\newblock An algorithm for finding hamilton cycles in random directed graphs.
\newblock {\em Journal of Algorithms}, 9(2):181--204, 1988.

\bibitem{HJ}
A.~W. Hales and R.~I. Jewett.
\newblock Regularity and positional games.
\newblock {\em Trans. Amer. Math. Soc.}, 106:222--229, 1963.

\bibitem{hefetz2009fast}
Dan Hefetz, Michael Krivelevich, Milo{\v{s}} Stojakovi{\'c}, and Tibor
  Szab{\'o}.
\newblock Fast winning strategies in maker--breaker games.
\newblock {\em Journal of Combinatorial Theory, Series B}, 99(1):39--47, 2009.

\bibitem{janson2011random}
Svante Janson, Tomasz Luczak, and Andrzej Rucinski.
\newblock {\em Random graphs}, volume~45.
\newblock John Wiley \& Sons, 2011.

\bibitem{krivelevich2009equitable}
Michael Krivelevich and Bal{\'a}zs Patk{\'o}s.
\newblock Equitable coloring of random graphs.
\newblock {\em Random Structures \& Algorithms}, 35(1):83--99, 2009.

\bibitem{west2001introduction}
Douglas~Brent West et~al.
\newblock {\em Introduction to graph theory}, volume~2.
\newblock Prentice hall Upper Saddle River, 2001.

\end{thebibliography}

\end{document}